\documentclass[journal]{IEEEtran}
\usepackage{amsfonts,cite}
\usepackage{graphicx}
\usepackage{amsfonts,color}
\usepackage[cmex10]{amsmath}
\usepackage{amssymb}
\usepackage{epsfig}
\usepackage{cite}
\usepackage{latexsym}
\usepackage{epstopdf}
\usepackage{amsmath, amsthm, amssymb}
\usepackage{verbatim}
\usepackage{subcaption}
\usepackage{amsmath}
\usepackage[utf8]{inputenc}
\usepackage{algorithmic}
\usepackage[algo2e]{algorithm2e}
\DeclareMathOperator{\Tr}{Tr}
\usepackage{algorithm}
\usepackage{enumitem}

\usepackage{soul}

\newtheorem{theorem}{Theorem}
\newtheorem{proposition}{Proposition}
\newtheorem{assumption}{Assumption}
\newtheorem{remark}{Remark}
\usepackage{stfloats}
\usepackage{blindtext}
\IEEEoverridecommandlockouts

\definecolor{orange}{RGB}{255,107,0}

\begin{document}
	\title{Cell-Edge Detection via Selective Cooperation and Generalized Canonical Correlation}
	\author{Mohamed Salah Ibrahim,~\IEEEmembership{Student~Member,~IEEE,}~Ahmed S. Zamzam,~\IEEEmembership{Member,~IEEE,}~Aritra~Konar,~\IEEEmembership{Member,~IEEE,} and~Nicholas~D.~Sidiropoulos,~\IEEEmembership{Fellow,~IEEE}
		\thanks{
		Mohamed~Salah~Ibrahim,~Aritra Konar~and~Nicholas~D.~Sidiropoulos are with the Department of Electrical and Computer Engineering, University of Virginia, Charlottesville, VA, 22904 USA (e-mails: \{salah,\ aritra,\ nikos\}@virginia.edu).  
	    Ahmed S. Zamzam is with the National Renewable Energy Laboratory, Golden, CO, 80401  USA (e-mail: ahmed.zamzam@nrel.gov)
	    }
	    }

\maketitle
\begin{abstract}
Improving the uplink quality of service for users located around the boundaries between cells is a key challenge in LTE systems. Relying on power control, existing approaches throttle the rates of cell-center users, while multi-user detection requires accurate channel estimates for the cell-edge users, which is another challenge due to their low received signal-to-noise ratio (SNR). Utilizing the fact that cell-edge user signals are weak but common (received at roughly equal power) at different base stations (BSs), this paper establishes a connection between cell-edge user detection and generalized canonical correlation analysis (GCCA). It puts forth a GCCA-based method that leverages selective BS cooperation to recover the cell-edge user signal subspace even at low SNR. The cell-edge user signals can then be extracted from the resulting mixture via algebraic signal processing techniques. The paper includes theoretical analysis showing why GCCA recovers the correct subspace containing the cell-edge user signals under mild conditions. The proposed method can also identify the number of cell-edge users in the system, i.e., the common subspace dimension. Simulations reveal significant performance improvement relative to various multiuser detection techniques. Cell-edge detection performance is further studied as a function of how many / which BSs are selected, and it is shown that using the closest three BS is always the best choice.
\end{abstract}


\section{introduction}
\IEEEPARstart{P}{roviding} high data rates to users located at the boundaries between cellular coverage areas
constitutes a major concern in the current 4G system~\cite{dahlman20134g} and the emerging 5G networks~\cite{andrews2014will}. Even with advanced technologies such as multiple-input-multiple-output (MIMO) and orthogonal frequency division multiplexing (OFDM) \cite{paulraj2004overview,jiang2007multiuser} in place, nomadic users who are close to the cell edge are still prone to suffer from significant performance degradation \cite{you2011cell,khan2016downlink}.  

Owing to the fact that the received signal power exhibits an inverse relationship with the propagation distance, cell-edge user terminals experience high path-loss that eventually results in a severe performance degradation~\cite{boudreau2009interference}. This effect becomes more pronounced when mobile systems operate at higher radio frequency, as expected in the future~\cite{niu2015survey}, thereby rendering the cell-edge user detection problem even more difficult. A variety of techniques, ranging from multi-user detection~\cite{verdu1998multiuser}, user scheduling ~\cite{chayon2017enhanced}, power control~\cite{gochev2013improving}, cooperative communication~\cite{kumar2008throughput,khan2016downlink}, and interference mitigation~\cite{ramamurthi2009cutting,borah2020effect,ghaffar2011interference} have been proposed as possible candidates for tackling this problem. 

Though optimal, the maximum likelihood detector (MLD)~\cite{verdu1998multiuser,proakis2001digital} requires solving an NP--hard combinatorial problem with computational complexity that grows exponentially with the number of users, thereby precluding its use in practical multi-antenna systems. The so-called sphere decoder~\cite{vikalo2003expected} (SD) is a near-optimal detector that can attain the MLD performance with lower complexity. However, it has been proven that its average complexity remains exponential~\cite{jalden2004exponential}. In the low to moderate signal-to-noise ratio (SNR) regime, semi-definite relaxation (SDR) based methods~\cite{ma2002quasi,wiesel2005semidefinite} can yield performance comparable to SD in polynomial time; yet their complexity remains {unaffordable} in terms of practical implementation at present~\cite{luo2010semidefinite}.
  
Whereas equalization-based detectors such as zero-forcing (ZF) and minimum mean square error (MMSE)~\cite{madhow1994mmse} exhibit substantially lower complexity compared to MLD, SD and SDR, their bit error rate (BER) performance often suffers severe degradation, especially at low SNR. The performance of both ZF and MMSE detectors can be further enhanced~\cite{wolniansky1998v} by using successive interference  cancellation (SIC), or decision  feedback (DF)~\cite{moshavi1996multi,duel1993decorrelating}, which rely on iteratively eliminating the strong (cell-center) user signals once  they  are  decoded. While one can also resort to joint detection using base station cooperation~\cite{khattak2008base}, this often degrades the performance because of near-far effects and the inaccurate channel estimates of cell-edge users.  

A major issue with all of the aforementioned detectors is that their performance is dependent on the availability of accurate channel state information of all users. This may be possible for users that are close to their serving base station (cell-center), and hence, reliable detection of such users can be  guaranteed. On the other hand, owing to high path-loss, the signals of cell-edge (weak) users are received at low SNR, which degrades the quality of their channel estimates. This together with inter- and intra-cell interference have a deleterious impact on the BER performance of the cell-edge users~\cite{borah2020effect}. 

Power control~\cite{gochev2013improving} and/or scheduling algorithms~\cite{chayon2017enhanced,xu2013joint} have proven to be successful in significantly improving the quality of service (QoS) of cell-edge users. This, however, comes at the expense of throttling the  cell-center user rates, and consequently the overall system throughput. Moreover, the frequent mobility-induced hand-off of cell-edge users renders their detection task even more challenging~\cite{agarwal2014qos}.

The shortcomings of the prevailing approaches motivate the following question: \emph{Does there exist a low-complexity method that can provide reliable detection of the cell-edge users without knowing their channels or sacrificing the performance of cell-center users by resorting to power control and scheduling techniques?} 

This paper provides an affirmative answer to this question, by proposing an unsupervised learning-based method that leverages {\em selective} base station cooperation to recover cell-edge users signals at low SNR subject to strong inter- and intra- cell interference. Relying on fact that cell-edge users are located at approximately equal distances from different base stations, and hence their received signals are weak but {\em common} (meaning: they are received at low but roughly equal power at different base stations), it shows that reliable detection is possible via (generalized) canonical correlation analysis (G)CCA~\cite{kettenring1971canonical} under mild conditions. 

While base station cooperation~\cite{khattak2008distributed} has been considered before for several tasks such as multi-cell cooperative beamforming~\cite{karakayali2006network}, coordinated power control~\cite{rao2007reverse}, coordinated scheduling~\cite{das2006interference}, and inter-cell interference mitigation~\cite{balachandran2010nice,marsch2011uplink}, cooperation here is utilized for a completely different purpose: as a means for cell-edge user detection, via GCCA. This work adds to the growing list of applications of (G)CCA  in several areas in signal processing and wireless communications, ranging from array processing~\cite{ge2009does}, direction-of-arrival (DoA) estimation~\cite{wu1994music}, radar anti-jamming~\cite{bai2005radar}, blind source separation~\cite{li2009joint,borga2001canonical,liu2007analysis,bertrand2015distributed}, speech processing~\cite{arora2014multi},  multi-view learning~\cite{arora2014multi,chen2019graph}, and more recently in cell-edge user detection using two BSs~\cite{salah2019}. Efficient algorithms have been recently developed for handling large scale GCCA~\cite{FuGCCA,FuMAXVAR,KanSUMCOR}.

Our contributions in this paper are as follows: 
\begin{itemize}
	\item We extend our previous work~\cite{salah2019}, which proposed using classical two-view CCA to detect cell-edge user signals in a cellular network with two cells, to the more general setting.  That is, we consider a scenario with $ L $ cells / views ($ L > 2 $) which involves GCCA~\cite{carroll1968generalization,horst1961generalized} as opposed to CCA for $ L = 2 $. We propose a two-stage approach that uses cooperation to jointly detect cell-edge users signals without prior knowledge of their channel state information. In particular, we first consider the so-called MAXVAR formulation of GCCA~\cite{carroll1968generalization}, and show that it  yields the range space of the cell-edge user signals. We present identifiability conditions under which the common subspace can be recovered. While identifiability conditions for the common subspace of two views have been obtained in~\cite{salah2019}, the conditions we provide here for the general case are more relaxed. 
	Upon identifying the subspace comprising the cell-edge users signals via GCCA, we utilize the (R)ACMA~\cite{van1997analytical} algorithm, which exploits the finite alphabet constraint of the user transmitted signals to retrieve the original cell-edge user signals from the resulting mixture. Fortunately, both MAXVAR GCCA and RACMA admit relatively simple algebraic solution via eigenvalue decomposition. This renders our approach computationally favorable in practice, because the proposed method for solving the cell-edge problem is tantamount to solving two eigenvalue decomposition problems. 
	  
	\item  We present an elegant theoretical analysis which shows that GCCA can reliably estimate the common subspace in the presence of thermal noise and cross interference from users in adjacent cells, under realistic assumptions on the SNR of the different users. 
	
	\item We provide an elegant GGCA strategy that can be used to classify users as cell-edge or cell-center, thereby determining the correct dimension of the common subspace. 
	  
	\item To showcase the effectiveness of our proposed method for cell-edge user detection, we provide a comprehensive suite of simulations that employs a realistic path-loss model from the 3GPP $38.901$ standard. Experiments reveal that our approach attains a considerable improvement in the BER at low SNR under realistic levels of inter-cell interference and dense scenarios with a large number of cell-center users. We compare our proposed method with our previous CCA-based one and different multi-user detection techniques including ZF-SIC and MMSE-SIC which assume perfect knowledge of the cell-center user channels. We show that the proposed GCCA method achieves significant reduction in the BER compared to all baselines. Moreover, our simulations show that using GCCA with the three closest BSs always yields the best detection performance for the cell-edge users. That is, not only does using the three closest BSs always improves the results of using the two closest BSs; but also that using more than the three closest BSs never helps -- neither of which was obvious {\em a priori}.
\end{itemize}
     
A preliminary version of part of the results in this paper has been submitted to the IEEE Global Communications Conference (GLOBECOM) 2020. Relative to the conference submission, this journal version includes i) two theoretical results and their detailed proofs, ii) a new section showing how the number of cell-edge users can be accurately estimated via the proposed method, which is important in practice, and iii) extensive numerical results with more realistic simulated scenarios and thorough discussion.     
     
\subsection{Paper Organization} 	
The outline of this paper is as follows. After briefly reviewing (G)CCA in Section \ref{GCCA_Overview}, Section \ref{Sys_Model} defines the problem statement and highlights the major limitations of the prior cell-edge user detection methods. The proposed detector and the main results are presented in Section \ref{PD}. Then, numerical simulations are provided in Section \ref{Simu}. Conclusions are drawn in Section \ref{Conc}. Long proofs and derivations are relegated to the Appendix. 

\subsection{Notation}
In this work, we use upper and lower case bold letters to denote matrices and column vectors, respectively. For any general matrix $ {\bf N} $, we use ${\bf N}^T$,  ${\bf N}^H$, ${\bf N}^{-1}$, ${\bf N}^{\dagger}$ and $ \Tr({\bf N}) $ to denote the transpose, the conjugate-transpose, the inverse (when it exists), the pseudo-inverse, and the trace of $ {\bf N} $, respectively. $ {\bf N}(:,m) $ denotes the $m$-th column of ${\bf N}$ (MATLAB notation). 
Furthermore, $\mathtt{Re}\{{\bf N}\}$ and $\mathtt{Im}\{{\bf N}\}$ extract the real part and the imaginary part of ${\bf N}$, respectively. Scalars are represented in the normal face, while calligraphic letters are used to denote sets. $ \lVert.\rVert_2 $ and $  \lVert.\rVert_F  $ denote the $ \ell_2 $-norm and the Frobenius norm, respectively. Finally, $ {\bf I}_N $ and $ {\bf 0}_{N \times M} $ denote the $ N \times N $ identity matrix and the $ N \times M $ zero matrix, respectively.

\section{Preliminaries}\label{GCCA_Overview}
Consider $L$ data sets $\{{\bf Y}_\ell \in \mathbb{R}^{M_\ell \times N}\}_{\ell = 1}^L$, where $ {\bf y}^{(n)}_\ell :=  {\bf Y}_\ell(:,n) $ is an $M_\ell$-dimensional feature vector that defines the $\ell$-th view of the $n$-th sample, $\forall n \in \mathcal{N} := \{1,\cdots,N\}$ and $\forall \ell \in \mathcal{L} := \{1,\cdots,L\}$. Without loss of generality, assume that all per-view data vectors $\{{\bf y}^{(n)}_\ell\}_{n=1}^{N}$ are zero-mean, otherwise the sample mean can be subtracted as a pre-processing step. 
While single-view analysis techniques, i.e., $L = 1$, like principal component analysis~\cite{wold1987principal} aim at extracting strong components from the given data matrix, multi-view analysis tools such as coupled matrix factorization (CMF) or canonical correlation analysis (CCA), seek to jointly analyze different views of the data. The main difference between CMF and (G)CCA lies in the optimization criterion: whereas CMF uses a data fitting (usually: least squares) criterion, (G)CCA is based on a ``differential'' criterion that forces it to zoom in only on what is common between the different views. If one of the views includes a very strong component that is absent from the other view(s), a least squares CMF formulation can still be obliged to represent that component. (G)CCA, on the other hand, owing to its use of a differential (balancing) criterion, can ignore principal components no matter how strong they are, as long as they are not common. For instance, the two-view CCA, i.e., $L = 2$, looks for two low-dimensional subspaces ${\bf Q}_1 \in \mathbb{R}^{M_1 \times K_c}$ and ${\bf Q}_2 \in \mathbb{R}^{M_2 \times K_c}$ with $K_c \ll \min \{N,M_\ell\}$, such that the distance between the linear projections of the received signals ${\bf Y}_1$ and ${\bf Y}_2$ onto these subspaces is minimized. From an optimization perspective, the distance-minimization formulation of the two-view CCA can be expressed as~\cite{hardoon2004canonical}
\begin{subequations}\label{dist_mini_CCA}
	\begin{align}
	&\underset{{\bf Q}_1,{\bf Q}_2}{\min}~  \| {\bf Y}_1^T{\bf Q}_1 - {\bf Y}_2^T{\bf Q}_2 \|^2_F\\
	& \text{s.t.} \quad~ {\bf Q}_\ell^T{\bf Y}_\ell{\bf Y}_\ell^T{\bf Q}_\ell = \textbf{{I}},~\forall \; \ell = \{1,2\}
	\end{align}
\end{subequations} 
The columns of ${\bf Q}_\ell$ are called the canonical components of the $\ell$-th view. Problem~\eqref{dist_mini_CCA} can be optimally solved via generalized eigenvalue decomposition~\cite{hotelling1936relations,borga2001canonical}.
To generalize problem~\eqref{dist_mini_CCA} to consider the case of multiple views $(L \geq 3)$, it is natural to adopt a pair-wise matching criterion~\cite{carroll1968generalization}. That is, we consider the optimization problem
\begin{subequations}\label{SUM-CORR}
	\begin{align}
	&\underset{\{{\bf Q}_\ell\}_{\ell=1}^{L}}{\min}~ \sum_{\ell = 1}^{L-1}\sum_{\ell' > \ell}^{L} \| {\bf Y}_\ell^T{\bf Q}_\ell - {\bf Y}_{\ell'}^T{\bf Q}_{\ell'} \|^2_F\\
	& \text{s.t.} \quad~ {\bf Q}_\ell^T{\bf Y}_\ell{\bf Y}_\ell^T{\bf Q}_\ell = \textbf{{I}},~\forall \; (\ell,\ell') \in \mathcal{L}.
	\end{align}
\end{subequations}
Problem~\eqref{SUM-CORR} is referred to as the sum-of-correlations (SUMCOR) generalized CCA~\cite{carroll1968generalization}. Although SUMCOR is known to be NP--hard in its general form \cite{rupnik2013comparison,zhang2011towards}, several efficient and scalable algorithms have been developed to obtain high-quality approximate solutions~\cite{KanSUMCOR,rupnik2013comparison,zhang2011towards}. 

Instead of minimizing the  distance between the reduced-dimension views, another formulation seeks a low-dimensional common latent representation, namely ${\bf G} \in \mathbb{R}^{N \times K_c} $, of the different views. This leads to the so-called maximum-variance (MAXVAR) GCCA formulation, which is given by
\begin{subequations}\label{MAXVAR}
	\begin{align} \label{MAXVARa}	
		&\underset{\{{\bf Q}_\ell\}_{\ell=1}^{L},{\bf G}}{\min}~ \sum_{\ell = 1}^{L} \| {\bf Y}_\ell^T{\bf Q}_\ell - {\bf G} \|^2_F \\
		& \text{s.t.} \quad~ {\bf G}^T{\bf G} = {\bf I}
	\end{align}
\end{subequations} 
Although both SUMCOR and MAXVAR aim at finding highly-correlated reduced-dimension views, and their solutions can be shown to coincide in the special case of $L=2$ (CCA), they are generally different for $L>2$ (GCCA). While problem~\eqref{MAXVAR} introduces an additional $NK_c$ variables compared to SUMCOR, it replaces the multiple constraints in~\eqref{SUM-CORR} with a single orthonormality constraint on the matrix ${\bf G}$. Fortunately, the MAXVAR GCCA formulation admits algebraic solution via eigenvalue decomposition. To see this, one can fix ${\bf G}$ and solve~\eqref{MAXVAR} with respect to ${\bf Q}_\ell$'s. Then, upon assuming that the matrices $\{{\bf Y}_\ell\}_{\ell = 1}^{L}$ are full row rank, it follows that ${\bf Q}_\ell^* = ({\bf Y}_\ell{\bf Y}^T_\ell)^{-1}{\bf Y}_\ell{\bf G}$. Substituting back ${\bf Q}_\ell^*$ in~\eqref{MAXVAR} and expanding the cost function in~\eqref{MAXVARa}, one can recast~\eqref{MAXVAR} as 
	 \begin{subequations}\label{ED_MAXVAR}
	 	\begin{align}
	 	&\underset{{\bf G}{\in \mathbb{R}^{N\times K_c}}}{\max}~ \text{Tr}({\bf G}^T{\bf A}{\bf G}) \\
	 	& \text{s.t.} \quad~ {\bf G}^T{\bf G} = {\bf I}
	 	\end{align}
	 \end{subequations} 
where $\bf{A} $ is defined as ${\bf A} := \sum_{\ell = 1}^{L}{\bf Y}_\ell^T({\bf Y}_\ell{\bf Y}_\ell^T)^{-1}{\bf Y}_\ell$. It can be easily seen that~\eqref{ED_MAXVAR} is nothing but an eigenvalue problem with the optimal solution ${\bf G}^\star$ being the first $K_c$ principal eigenvectors of the matrix ${\bf A}$~\cite{van1983matrix}. In what follows, we will focus on the MAXVAR formulation to show how GCCA relates to the problem of cell-edge user detection in multi-cell, multi-user MIMO systems. 

\section{System Model}\label{Sys_Model}
Consider an uplink transmission scenario in a cellular network with $ L $ regular hexagonal cells -- each cell has a base station (BS) located at its center, as shown in Fig.~\ref{SM1}. The $ \ell $-th BS is equipped with $ M_\ell $ antennas, and serves $ K_\ell $ single-antenna users. Let $ K_{e_\ell} $ denote the number of cell-edge users served by the $ \ell $-th BS, with $ K_{e_\ell} < K_\ell, \forall \; \ell \in \mathcal{L}: = \{1, \cdots, L\}$. The uplink channel vector representing the small-scale fading between the $ k $-th user located in the $ j $-th cell and the $ \ell $-th BS is given by $ {\bf z}_{\ell kj} \in \mathbb{C}^{M_\ell}$. The entries of $ {\bf z}_{\ell kj} $ are independent identically distributed (i.i.d.) complex Gaussian random variables with zero mean and variance $ 1/ M_\ell $. This corresponds to a favorable propagation medium with rich scattering.  
The coefficient that accounts for the signal attenuation due to distance (path-loss) between the $ \ell $-th BS and the $ k $-th user in the $ j $-th cell is given by $ \alpha_{\ell kj} \in \mathbb{R}$. Accordingly, the overall uplink channel vector is given by  
\begin{align}\label{PS1}
	{\bf h}_{\ell kj}  = \sqrt{\alpha_{\ell kj}}{\bf z}_{\ell kj}
\end{align}
Throughout this work, we assume that the channel vectors are \emph{not} known a priori at the BSs. 
\begin{figure}[!t]
	\centering
	\includegraphics[width=3in]{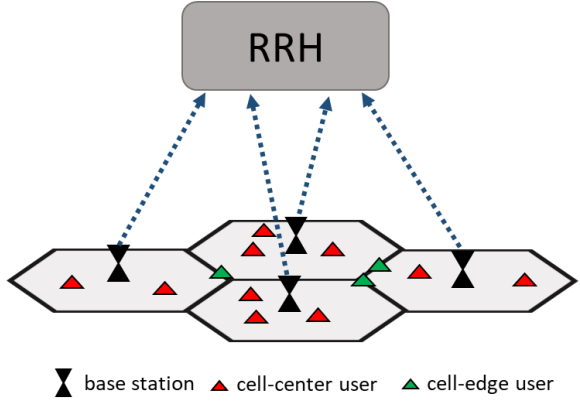}
	\caption{System model}
	\label{SM1}
\end{figure}

\subsection{Uplink Transmission}
The considered users in the system are assumed to be allocated the same time-frequency resource. Also, assume that all user transmissions are heard at all BSs, thereby introducing both intra- and inter-cell interference at each BS. Also, all user transmissions are assumed to be synchronized at the BSs (this assumption can be lifted; see below and in \cite{salahtwc}). Define ${\bf x}_{kj} \in \mathbb{R}^N$ as the binary vector transmitted by the $ k $-th user in the $  j$-th cell. While we assume binary phase shift keying (BPSK) transmissions, our approach also works for other modulation schemes but with slight variations in the second stage (we provide more details in subsection \ref{RACMA}). The received signal, $ {\bf Y}_\ell \in \mathbb{C}^{M_\ell \times N} $, at the $ \ell $-th BS is given by
\begin{equation}\label{UT1}
	{\bf Y}_\ell = \sum_{j = 1}^{L}\sum_{k = 1}^{K_j} \sqrt{p_{kj}}{\bf h}_{\ell kj}{\bf x}_{kj}^T + {\bf N}_{\ell}
\end{equation}
where $ {\bf h}_{\ell kj} $ is the uplink channel response vector as defined in~\eqref{PS1}, $ p_{kj} $ is the transmitted signal power of the $ k $-th user in the $ j $-th cell. 
The term $ {\bf N}_\ell \in \mathbb{C}^{M_\ell \times N}$  
contains i.i.d. entries with zero mean and variance $ \sigma^2/N $, i.e., $ \mathbb{E}[ {\bf N}_\ell{\bf N}^H_\ell ] = \sigma^2{\bf I} $. Throughout this work, we assume that neither scheduling algorithms nor power control is  employed. In other words, all users are always active, and all users are assigned the same transmission power, i.e., $ p_{kj} = p, \forall\; k,j $. Scheduling and power control algorithms can still be employed on top of the proposed framework for additional traffic shaping and other system considerations.  

In this paper, we assume that all BSs are connected to a remote radio head (RRH) via backhaul links which can be either microwave links or high speed optical fiber cables~\cite{khattak2008distributed}. Each BS forwards its received signal to the RRH. Although base station cooperation has been considered in several earlier papers ~\cite{karakayali2006network,rao2007reverse,das2006interference,balachandran2010nice,marsch2011uplink}, it is adopted here for an entirely different purpose. That is, we exploit the joint processing of the BS signals at the RRH to provide reliable detection of cell-edge users whose signals are received at low SNR without knowledge of any of the user channels. One key challenge for all the BS cooperation techniques in the literature is time synchronization~\cite{khattak2008distributed} of the received signal at the RRH. Even though all prior works assumed perfect synchronization at the RRH~\cite{khan2016downlink}, we recently developed a low-complexity CCA-based algorithm that can handle (lack of) synchronization for two BSs~\cite{salahtwc}, and can be easily modified to deal with the multi-cell case. 

\subsection{Prior Art: Limitations and Challenges}
We now provide a brief discussion of the limitations of the prior art used to detect cell-edge user signals. To this end, it is convenient to write \eqref{UT1} as
\begin{align}\label{CE1}
	{\bf Y}_\ell = {\bf H}_{\ell p_\ell}{\bf X}_{p_\ell}^T + {\bf H}_{\ell e_\ell}{\bf X}^T_{e_\ell}+ \sum_{j \neq \ell}^L{\bf H}_{\ell j}{\bf X}_{j}^T
	+ {\bf N}_\ell   
\end{align}
where we collect the transmitted signals of the cell-center users and the cell-edge users served by the $ \ell $-th BS in the matrices $ {\bf X}_{p_\ell} \in \mathbb{R}^{N \times (K_\ell - K_{e_\ell})}$ and $ {\bf X}_{e_\ell} \in \mathbb{R}^{N \times K_{e_\ell}} $, respectively, and the transmitted signals by the users served by the $ j $-th BS in $ {\bf X}_{j} \in \mathbb{R}^{N \times K_j} $.
Further, the matrices $ {\bf H}_{\ell p_\ell} \in \mathbb{C}^{M_\ell \times (K_\ell - K_{e_\ell})}$, $ {\bf H}_{\ell e_\ell} \in \mathbb{C}^{M_\ell \times  K_{e_\ell}}$ and $ {\bf H}_{\ell j} \in \mathbb{C}^{M_\ell \times K_j}$ hold on their columns the respective channel vectors. Note that we absorbed the transmitted signal power $ p $ of each user in its channel vectors. 

The goal is to recover the cell-edge user signals, $ {\bf X}_{e_\ell} $, from the received signal $ {\bf Y}_\ell $. The traditional approach to recover $ {\bf X}_{e_\ell} $ is to first estimate all user channels via transmitting orthogonal pilots, and then employ the ZF or MMSE detector to decode cell-edge user signals using their estimated channel. This approach usually fails to provide reasonable performance due to the effect of intra-cell interference (transmissions of strong cell-center users), the effect of inter-cell interference (transmissions from users in other cells) and noise. The signals of cell-edge users are consequently received at very low signal to interference plus noise ratio (SINR), which causes high uncertainty in their channel estimates, which in turn seriously degrades their detection performance. One workaround is to use ZF or MMSE followed by successive interference cancellation to decode and subtract the cell-center user signals, thereby mitigating / eliminating the intra-cell interference effect. While this approach can slightly improve the detection performance, the inter-cell interference and channel estimation errors still cause severe performance degradation.

Another potential solution that mitigates the inter-cell interference effect~\cite{gesbert2010multi}, and can indeed enhance the detection of such users, is to use power control and/or scheduling algorithms together with BS cooperation techniques~\cite{zhang2010cooperative}. However, this comes at the expense of throttling the transmission of cell center users, and hence, it also severely degrades the overall system throughput. 

In the forthcoming section, we will present a two-stage learning-based approach that leverages BS cooperation to reliably identify cell-edge user signals without knowing their channels, and without resorting to either power control or scheduling.
\section{Proposed Detector and Identifiability Analysis}\label{PD}
The cell-edge users are located far but at roughly equal distances from different BSs. In other words, if we use the distance-power relationship, their received signals are weak but \textit{common} to multiple BSs, i.e., their signals are received at relatively equal power at different BSs. We will show how GCCA can efficiently recover the cell-edge users' signal range space at low SNR, even if they are buried under strong intra- and inter-cell interference. Notice that, owing to the broadcast nature of the wireless medium, all user transmissions are (over)heard, albeit weakly, at all BSs. Hence all user signals are, in principle, common. However, we use the phrases ``common'' for cell-edge users versus ``private'' for cell-center users to reflect the power (im)balance of different users. That is, cell-center users signals are received at very high SNR, e.g., $[20,30]$ dB, at their serving BS, and very low SNR, eg. $[-10,-30]$ dB at non-serving BSs. On the other hand, cell-edge user signals are received at low but roughly equal SNR, e.g., $[3,5]$ dB at multiple BSs. 

From the geometry of the hexagonal cells shown in Fig.~\ref{SM1}, one can argue that a user can be common to two or three BSs, i.e., it can be located at relatively equal distance from two or three BSs. For example, a user located on the left corner of the common edge between the top and bottom cells in Fig.~\ref{SM1} is common to the BSs in these two cells and the one on the left. Based on this fact, we will design a detector that can recover cell-edge user transmitted signals from the signal received at three BSs. The case of more than three will also be considered in the simulations section.

We will first consider the noiseless case to find the identifiability conditions required to recover the cell-edge user signals. Identifiability is very important as it provides sufficient conditions under which the recovery of the cell-edge user signals via (G)CCA is guaranteed under ideal (noiseless) conditions. Whereas we have derived identifiability conditions in the case of two BSs~\cite{salah2019}, it turns out that the conditions for three BSs are more relaxed (details will be provided in the next subsection). 

\subsection{Noiseless Case}  
Let $ K_c= \sum_{j = 1}^{L} K_{e_j}$ denote the total number of cell-edge users. 
Assume that all cell-edge users are located around the intersection point of the three hexagonal cells. Thus equation~\eqref{CE1}, with $ L = 3 $, can be rewritten as
\begin{align}\label{NC1}
{\bf Y}_\ell = {\bf H}_{\ell p_\ell}{\bf X}_{p_\ell}^T + {\bf H}_{\ell c}{\bf X}^T_{c}+ \sum_{j \neq \ell}^L{\bf H}_{\ell p_j}{\bf X}_{p_j}^T
+ {\bf N}_\ell   
\end{align}
where the subscripts $ `p_\ell` $ and $ `c` $ stand for private to the $ \ell $-th BS and common to all base stations, respectively. The matrices $ {\bf X}_c \in \mathbb{R}^{N \times Kc} $ and $ {\bf X}_{p_j} \in \mathbb{R}^{N \times (K_j - K_{e_j})} $ hold the transmitted signals of the cell-edge users and cell-center users in the $ j $-th cell, respectively. Accordingly, $ {\bf H}_{\ell c} \in \mathbb{C}^{M_\ell \times  K_c}$ and $ {\bf H}_{\ell p_j} \in \mathbb{C}^{M_\ell \times( K_j - K_{e_j})}$ hold on their columns the corresponding channel vectors. Further, we define $ {\bf E}_\ell := \sum_{j \neq \ell}^L{\bf H}_{\ell p_j}{\bf X}_{j p_j}^T
+ {\bf N}_\ell $ to denote the summation of inter-cell interference and noise at the $ \ell $-th BS. Thus,~\eqref{NC1} can be rewritten as
\begin{align}\label{NC2}
{\bf Y}_\ell = {\bf H}_{\ell p_\ell}{\bf X}_{p_\ell}^T + {\bf H}_{\ell c}{\bf X}^T_{c}+ {\bf E}_\ell   
\end{align} 

As a pre-processing step, let us transform the received signals to the real domain by constructing the matrix $ \overline{\bf Y}_\ell := [{\bf Y}_\ell^{(r)};{\bf Y}_\ell^{(i)}]  \in \mathbb{R}^{2M_\ell \times N}$, where  $ {\bf Y}_\ell^{(r)} =      \mathtt{Re}\{{\bf Y}_\ell\} $ and $ {\bf Y}_\ell^{(i)} = \mathtt{Im}\{{\bf Y}_\ell\}  $. Similarly, form $ \overline{\bf H}_{\ell p_\ell} := [{\bf H}_{\ell p_\ell}^{(r)};{\bf H}_{\ell p_\ell}^{(i)}]  \in \mathbb{R}^{2M_\ell \times( K_\ell - K_{e_\ell})}$, $ \overline{\bf H}_{\ell c} := [{\bf H}_{\ell c}^{(r)};{\bf H}_{\ell c}^{(i)}]  \in \mathbb{R}^{2M_\ell \times K_c}$ and $ \overline{\bf E}_\ell = [{\bf E}^{(r)}_\ell;{\bf E}^{(i)}_\ell] \in \mathbb{R}^{2M_\ell \times N}$. Therefore,~\eqref{NC2} can be equivalently expressed as
\begin{equation}\label{MV2}
\overline{\bf Y}_\ell = \overline{\bf H}_{\ell p_\ell}{\bf X}_{p_{\ell}}^T + \overline{\bf H}_{\ell c}{\bf X}^T_c + \overline{\bf E}_\ell.   
\end{equation}
Recall that~\eqref{MV2} is nothing but a more compact form of~\eqref{UT1} in the real domain. Our goal now is to recover the cell-edge user signals ${\bf X}_c$ given $\{\overline{\bf Y}_\ell\}_{l=1}^{L}$. We will start by showing how the solution of the MAXVAR GCCA formulation~\eqref{MAXVAR} is related to the column space of the cell-edge user signals, and then we will explain how the original signals can be recovered from the given solution. Upon defining the matrix ${\bf V}^{(L)} \in \mathbb{R}^{(L-1)N \times ((L-1)K_c + \sum\limits_{\ell=1}^L(K - K_{e_\ell}) )}$ as follows,
\[
{\bf V}^{(L)} = \begin{bmatrix} 
{\bf X}_{p1} & -{\bf X}_c & {\bf X}_{p2} & & \\
\vdots &  & \ddots & \ddots & \\
{\bf X}_{p1} &  &   &    -{\bf X}_c    & {\bf X}_{pL} 
\end{bmatrix}
\]
we have the following result.  
\begin{theorem}
	In the case where $ \overline{\bf E}_\ell = 0 $, if the matrix ${\bf W}_\ell := [\overline{\bf H}_{\ell c},\overline{\bf H}_{\ell p_\ell}] \in \mathbb{R}^{2M_\ell \times (K_c +  K_\ell  - K_{e_\ell})}$ and the  matrix ${\bf V}^{(L)} \in \mathbb{R}^{(L-1)N \times ((L-1)K_c + \sum\limits_{\ell=1}^L(K - K_{e_\ell}) )}$ are full column rank, 
    then the optimal solution $ {\bf G}^{\star} $ of problem~\eqref{MAXVAR} is given by $ {\bf G}^{\star} = {\bf X}_c {\bf F} $, where $ {\bf F} $ is a $ K_c \times K_c $ non-singular matrix. 
\end{theorem} 
\begin{proof}
	See ~\cite{sornsen2020} which offers a comprehensive identifiability analysis of GCCA for general $L$.  
\end{proof}

\begin{remark}
    The full column rank condition on ${\bf W}_\ell$ requires that: $i)$ the number of antennas at each base station is greater than or equal to half the number of users assigned to that base station, plus any cell-edge users associated with the other two base stations; and $ii)$ the channel vectors of different users to be linearly independent. The first requirement is supported by massive MIMO technology that aims at equipping the base station with hundreds of antennas~\cite{larsson2014massive}. Further, because the user channel vectors can be assumed to be drawn from an absolutely continuous distribution (see~\eqref{PS1}), the latter condition is satisfied with probability one. The full column rank condition on ${\bf V}$ is related to the number of samples $N$ required to guarantee recovery. For $ L = 3$, we need: $i)$ the number of samples to be at least equal to the total number of cell-edge users plus one half the total number of cell-center users served by the three BSs; and $ii)$ the transmitted sequences of different users to be linearly independent. For finite alphabets, the two conditions are satisfied with very high probability for modest $N$, as the user transmissions are statistically independent. In addition, it can be easily seen that the requirement on the number of samples $N$ becomes less restrictive for $L > 2$ compared to the earlier results for the two-view case in ~\cite{salah2019}.
\end{remark}
 Theorem $1$ asserts that in an ideal scenario where the effect of inter-cell interference and thermal noise is negligible compared to the intra-cell interference, GCCA successfully recovers the subspace spanned by the cell-edge users signals, under mild conditions. We point out that such a scenario can arise in practice, especially if all cell-center users are close to their serving BS. 
 An interpretation of the statement of Theorem $1$ is that if there exist several spatio-temporal signal views that contain very strong but different components (in our case here arising from the transmissions of each group of cell-center users) and very weak but common components (in our case here the received cell-edge user signals), then GCCA recovers the common components range space irrespective of the power of the individual components.

However, in practical deployment scenarios we cannot guarantee the above idealized assumptions. One of the main contributions of this paper is that it offers an analysis of GCCA performance in a realistic scenario with inter-cell interference and noise. This is coming up next. 

\subsection{Noisy Case}
We now provide analysis showing how cell-edge users signals can be identified when $ \overline{\bf E}_\ell \neq 0 $. In particular, we show that the signal subspace recovered by identifying the $K_c$ principal eigenvectors of ${\bf A}$ is indeed containing the cell-edge user transmitted messages. As shown earlier, this is equivalent to solving the MAXVAR GCCA problem.

Upon defining $ K_s = \sum_{\ell = 1}^L K_\ell $, let us write~\eqref{NC1} in a more compact form as
\begin{align}\label{IC1}
	{\bf Y}_\ell = {\bf H}_{\ell}{\bf X}^T + {\bf N}_\ell   
\end{align}
where $ {\bf H}_\ell = [{\bf H}_{\ell c},{\bf H}_{\ell p_1},\cdots,{\bf H}_{\ell p_L}] \in \mathbb{C}^{M_\ell \times K_s} $, and $ {\bf X} = [{\bf X}_c, {\bf X}_{p_1},\cdots,{\bf X}_{p_L} ] \in \mathbb{R}^{N \times Ks} $. Further, we can use~\eqref{PS1} to factor $ {\bf H}_{\ell}  = {\bf Z}_\ell {\bf P}^{1/2}_\ell$, where the columns of $ {\bf Z}_\ell $ are the channel vectors representing small scale fading between the $ k $-th user and the $ \ell $-th BS, for $ k \in \mathcal{K}_s := \{1, \cdots, K_s\} $. Accordingly, each entry of the diagonal matrix $ {\bf P}_\ell$ represents the corresponding received signal power that incorporates the transmitted power and the path-loss between each user and the $ \ell $-th BS. Thus,~\eqref{IC1} can be equivalently written as
\begin{align}\label{IC2}
	{\bf Y}_\ell = {\bf Z}_{\ell}{\bf P}_\ell^{1/2}{\bf X}^T + {\bf N}_\ell.  
\end{align}

\begin{assumption}[AS1]
	Assume that the matrices ${\bf Z}_\ell$ and $ {\bf C} := {\bf X} / \sqrt{N} $ are approximately orthonormal, i.e., ${\bf Z}^H_\ell {\bf Z}_\ell \approx {\bf I}_{M_\ell} $ for all $\ell \in \mathcal{L}$ and ${\bf C}^T {\bf C} \approx {\bf I}_{K_s} $.
\end{assumption}

\begin{remark}
	Recall that the matrices $ {\bf Z}_\ell $ contain i.i.d entries with zero mean and variance $ 1/M_\ell $, and hence, the approximate orthonormality assumption on $ {\bf Z}_\ell $ requires the number of base station antennas to be greater than the total number of users assigned to all base stations, and large enough for the sample average (inner product of different columns) to be close to the ensemble average ($\boldsymbol{0}$). This requirement on the number of antennas is supported by massive MIMO technology that aims at equipping base stations with hundreds of antennas. On the other hand, the approximate orthonormality of $ {\bf C} $ requires the sequence length to be greater than the total number of users and the columns of $ {\bf C} $ to be linearly independent. For finite alphabets, the latter condition is satisfied with very high probability for modest $N$.
\end{remark}
Let $\gamma_{k \ell}$ denote the received SNR of the $k$-th user at the $\ell$-th BS. Then, we define $r_{k\ell}$ as $$ r_{k\ell} := \frac{\gamma_{k\ell}}{\gamma_{k\ell}+1}$$ for all $k \in \mathcal{K}_s$ and $\ell \in \mathcal{L}$. For any user $k$, the effective SNR $\eta_k$ is defined as $\eta_k := \sum_{\ell =1}^{L} r_{k\ell}$. We will make use of the following assumption on the cell-edge users.

\begin{assumption}[AS2]
	For any cell-edge user $i$ and cell-center user $j$, $\eta_i > \eta_j$.  
\end{assumption}

\begin{remark}
	Empirically, the relation between the average received power $P_r$ and the distance is determined by the expression $P_r \propto d^{-\lambda}$ where $d$ and $\lambda$ denote the distance and the path loss exponent, respectively. The noise power at the receiver is given by $\sigma^2$. Then, the value of $r_{k\ell}$ as function of the distance between the user and the BS ($d_{k\ell}$) is given by
	\begin{align}
		r_{k\ell} = \frac{(d_{k\ell})^{-\lambda}}{(d_{k\ell})^{-\lambda} + \sigma^2/c}
	\end{align}
	where $c$ is constant that depends on the communication medium and the antenna characteristics. This function exhibits a sharp phase transition which means that the ratio $r_{j\ell}$ for cell-center users at other cells is almost zero if they are dropped up to certain distance from their serving BS such that their received SNR at their non-serving BSs is a few dBs below zero, while all the ratios $r_{i\ell}$ for cell-edge users at their adjacent BSs is close to one if their received SNR at those BSs is a few dBs above zero. 
\end{remark} 

Our main result is the following:
\begin{proposition}
	In the presence of inter-cell interference and additive noise, under assumptions (AS1) and (AS2), 
	the optimal solution $ {\bf G}^{\star} $ of problem~\eqref{MAXVAR} is given by $ {\bf G}^{\star} \approxeq {\bf X}_c{\bf P}$ where $ \bf P $ is any $K_c \times K_c$ non-singular matrix. 
\end{proposition}
\begin{proof}
	The proof is relegated to the Appendix.
\end{proof}
We have thus showed that MAXVAR GCCA identifies the \emph{range space} of the cell-edge users signals, under realistic conditions. We will next show how the original signals ${\bf X}_c$ can be unraveled from their range space, by exploiting their constellation/modulation structure. 
\subsection{RACMA Stage}\label{RACMA}
Given the subspace $ {\bf G}^* $, the problem of recovering the user signals  $ {\bf X}_c $ can then be posed as
\begin{subequations}\label{CSI2}
	\begin{align}
	&\underset{{\bf X}_c,{\bf F}}{\min}~ \| {\bf G}^{\star} - {{\bf X}}_c{\bf F} \|_F^2 \\
	&\text{s.t.} \quad {\bf X}_c{(i,j)} \in \{\pm 1\}
	\end{align}
\end{subequations}
where $ {\bf X}_c{(i,j)}$ represents the $(i,j)$-th entry of the matrix ${\bf X}_c$, for $i = 1, \cdots, N$ and $j = 1, \cdots, K_s$. Although problem~\eqref{CSI2} is known to be NP-Hard even if ${\bf F}$ is known, the  Real Analytical Constant Modulus Algorithm (RACMA) developed by van der Veen ~\cite{van1997analytical} provides a good algebraic solution which comes with certain identifiability guarantees. In particular, RACMA transforms~\eqref{CSI2} to a generalized eigenvalue problem. While the obtained solution is subject to both sign and permutation ambiguities, both of them can be resolved in practice by simply matching the preamble of each estimated signal with the identification sequence that is known {\it a priori} at the serving BS.

It is worth emphasizing that the proposed end-to-end detector of the cell-edge user signals only requires solving two generalized eigenvalue problems. Therefore, the overall computational complexity of our proposed method is dominated by the complexity of solving two generalized eigenvalue problems. This renders our approach favorable for practical implementation. 

We also point out that our proposed method works for other modulation schemes and even for analog signals. It is obvious that GCCA (first stage) can identify the common subspace irrespective of the modulation of the cell-edge user signals. The second stage exploits knowledge of the modulation to unravel the constituent signals from their range space. We used (R)ACMA~\cite{van1997analytical} for BPSK signals, but we can also use related methods, such as ACMA for higher-order PSK or even analog phase or frequency-modulated signals, and iterative algorithms such as SIC-ILS ~\cite{li2000blind} for higher-order QAM, which can also exploit Forward Error Control (FEC) codes to further improve the decoding accuracy. 

\subsection{{Choosing the common subspace dimension}}\label{UC}
Recall that in our analysis in the Appendix, to differentiate between cell-center users and cell-edge users, we assumed that cell-center users are dropped up to certain distance from their serving BS (see Remark 3). However, if all users are randomly dropped throughout the cells, then it is not obvious how to differentiate if a user has to be treated as a cell-center or a cell-edge user.
In other words, how to determine the common subspace dimension if the cell-center users are fully scattered within their cell. Note that underestimating the common subspace dimension can naturally lead to a performance degradation as we will see in the simulation section. To overcome such an issue, we propose a GCCA-based algorithm that can accurately estimate the number of cell-edge users (common subspace dimension), and hence, we can classify whether a user is cell-center or cell-edge. 

Exploiting the fact that a component that is common to three views should also be common to each pair of the three views, the common subspace dimension can be accurately estimated via checking the mean of correlation coefficients computed from the canonical components of each pair. Recall that $K_c \leq \min\{2M_\ell,N\}$, where $K_c$ is the number of canonical pairs that can be extracted using GCCA. Upon solving problem~\eqref{MAXVAR} and obtaining the solutions $\{{\bf Q}^\star_\ell\}_{\ell=1}^{3}$, we define the $i$-th correlation coefficient between views $j$ and $\ell$ as
\begin{equation}
	\rho^{(i)}_{\ell j} = {\bf Q}^T_\ell(:,i)\overline{\bf Y}_\ell\overline{\bf Y}^T_j{\bf Q}_j(:,i)
\end{equation}
$\forall \ell,j \in \mathcal{L}$ and $j > \ell$ and $i \in \{1,\cdots,K_c\}$. Afterwards, we compute the $i$-th average correlation coefficient as
$$ \rho^{(i)}_\text{avg} = \frac{1}{3}(\rho^{(i)}_{1 2}+\rho^{(i)}_{13}+\rho^{(i)}_{23}).$$
Then, we decide that the $i$-th canonical components (${\bf Q}_1(:,i),{\bf Q}_2(:,i),{\bf Q}_3(:,i)$) extract a common signal if $\rho^{(i)}_\text{avg}$ is greater than a certain threshold -- a reasonable choice of $\rho_\text{th}$ is $ 0.5 $.

\section{Experimental Results}\label{Simu}
In this section, we use realistic numerical simulations to assess the performance of the proposed GCCA approach. We consider a scenario with $ L = 4 $ hexagonal cells, each of radius $ \text{R} = 600 $ meters. The locations of cell-center users served by each BS are drawn uniformly at random within a distance less than $ d = 0.4\text{R} $ from their serving BS, unless stated otherwise. Cell-edge users, on the other hand, are located around the edges between base stations at distance between $0.95\text{R}$ and $1.05\text{R}$. Fig.~\ref{Simu_Scenario1} shows one simulated scenario where cell-center users and cell-edge users are colored in red and green triangles, respectively. The transmitted power of all users was set to $ 25 $ dBm  while the transmitted sequence length $ N $ was fixed to $ 800 $.  
All results  were averaged over $ 500 $ Monte Carlo trials. Additive white Gaussian noise was used with variance $ \sigma^2 $ so that the SNR is $ P_e / \sigma^2 $, where $ P_e $ is the average received power of cell-edge users. In fact, this enables us to evaluate SNR values required for cell-edge users to achieve a specific BER.
The uplink channel between the $k$-th user in the $j$-th cell and the $\ell$-th BS is modeled as 
\begin{equation}
	{\bf h}_{\ell kj}^H = \sqrt{\frac{1}{M}} \sum\limits_{n=1}^{N_p} \sqrt{\alpha_{\ell kj}^{(n)}}{\bf a}_r(\phi^{(n)})^H
\end{equation}
where $ N_p $ is the number of paths between the $ \ell $-th BS and the $ k $-th user in cell $ j $, $\forall \{\ell,j\} \in \mathcal{L} $ and $k \in [K_s]$. To compute the path gain, $ \alpha_{\ell kj}^{(n)}$, we use the path-loss model of the urban macro (UMa) scenario from Table $7.4.1-1$ in the 3GPP $ 38.901 $ standard, with the carrier frequency set to $2$ GHz,  $ \forall n,\ell,j,k $. Furthermore, all cell-center users are allowed to possibly have a line of sight (LOS) component to their serving BS according to the LOS probability expression for the UMa scenario in Table $7.4.2-1$ in the 3GPP $ 38.901 $ standard; however, all cell-edge users have only non-LOS components. The term $ {\bf a}_r(.) $ is the array response vector at the BS, and $ \phi^{(n)} \sim \mathcal{U}[-\pi,\pi]$ denotes the azimuth angle of arrival of the $ n $-th path. Assuming the BS is equipped with a uniform linear array, then
\begin{equation}
	{\bf a}_r(\theta) = [1,e^{ikd\cos(\phi)}, \cdots,e^{ikd(M-1)\cos(\phi)}]^T,
\end{equation}
where $i=\sqrt{-1}$, $ k = 2\pi/\lambda$, $ \lambda $ is the carrier wavelength and $ d = \lambda/2$ is the spacing between antenna elements.

To assess the efficacy of our approach, we implement the following approaches and use them as performance baselines.
\begin{itemize}
	\item {\bf MMSE / ZF with channel estimation}: the channels of all users are estimated via transmitting sequences of orthogonal pilots with length of $250$ each. Then, both the MMSE and ZF detectors are employed to decode the cell-edge user signals using their estimated channels.
	
	\item {\bf MMSE / ZF SIC RACMA with channel estimation}: the channels of the cell-center users associated with each BS are estimated first. Then, we use both the ZF-SIC and MMSE-SIC detectors to decode, and then subtract the re-encoded signals of the cell-center users at their serving BS. Afterwards, we apply RACMA~\cite{van1997analytical} on the residual signal to recover the cell-edge user signals. To guarantee fairness, since we assume BS cooperation, we feed RACMA with the residual signals from all BSs simultaneously.
	
	\item {\bf MMSE / ZF SIC RACMA Perfect}: similar to the previous baseline but with perfect knowledge of the cell-center user channels at their serving BS.
	
	\item {\bf CCA RACMA combined}: we use CCA to recover the range space of the cell-edge users from the nearest two BSs~\cite{salah2019}. Then, we apply RACMA to recover the cell-edge user signals from the resulting subspace.
\end{itemize}
\begin{figure}
	\centering
	\includegraphics[width=3in]{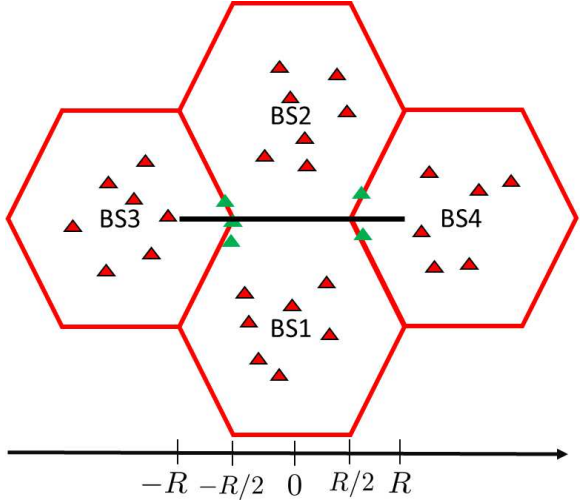}
	\caption{Snapshot from the simulated scenario}
	\label{Simu_Scenario1}
\end{figure}

In the first experiment, we consider a setup with $M_\ell = 12$ antennas and $K_\ell = 8$ single transmit antenna users, $~\forall \; \ell \in \mathcal{L}$. Considering the scenario shown in Fig. \ref{Simu_Scenario1}, we varied the x-location of one cell-edge user on the black edge between BS1 and BS2 from $x = -\text{R}$ to $x = \text{R}$,  while the locations of all other users are kept fixed during the experiment. At each value of $x$, we report the BER of the proposed GCCA approach using the three closest BSs, GCCA using all BSs, CCA using the closest two BSs and ZF-SIC RACMA with perfect cell-center users channels using the best BS. Note that when the user is located at $x = -\text{R}/2$, its received SNR is approximately equal to $ 3 $dB at BSs 1,2 and 3, according to our adopted path-loss model. As the user's location shifts towards the center ($x = 0$) of the black edge, its received SNR increases (as path-loss decreases) at BSs $1$, $2$, and $4$ while it decreases at BS $3$. When this user passes the center of the edge, the received SNR decreases again at BSs $1$ and $2$. On the other hand, when the user's location changes towards the very left corner of the black edge ($x = -\text{R}$), its received SNR automatically increases at BS $3$, while it decreases at BSs $1,2,4$.

As Fig.~\ref{Res_simscenario1} depicts, when the user is located at $x = -\text{R}/2$ (the user is at relatively equal distance from three BSs), the proposed GCCA using the three left BSs ($1$, $2$, and $3$) attains the minimum BER compared to GCCA using all BSs, CCA using the two closest BSs and ZF-SIC using the best BS. Similarly, when the user is located at $x = \text{R}/2$, the joint detection using the three BSs $1$, $2$, and $4$ gives the best performance. When the user is close to $x = \pm \text{R}/2$, GCCA using the three nearest BSs attains more than order of magnitude reduction in the BER relative to CCA RACMA and much more relative to ZF-SIC RACMA. 

On the other hand, as the location of the user moves towards the center of the cell-edge $(x = 0)$, the detection performance of CCA using the two BSs $1$, $2$ improves gradually until it reaches its best at the origin (minimum path-loss and maximum received SNR), and then it decreases again as shown in Fig.~\ref{Res_simscenario1}. This happens as the received SNR of the user's signal at BS $1$ and BS $2$ becomes higher, and considerable discrepancy becomes evident between the received SNR at BSs $1$ and $2$ and at BS $3$. Note that GCCA using the best three BSs always yields the minimum BER when the user's location is in the interval $[-\text{R}/2,\text{R}/2]$.

\begin{figure}
	\centering
	\includegraphics[width=3in]{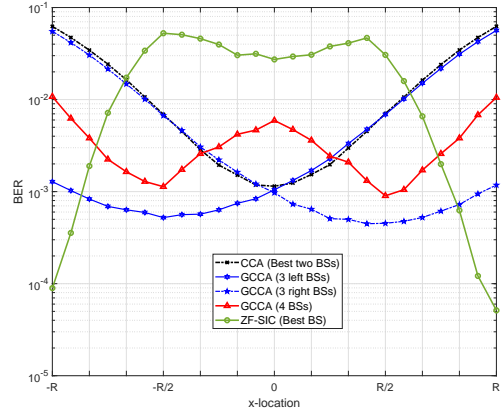}
	\caption{BER vs. cell-edge user location: GCCA using 3 closest BSs is always better}
	\label{Res_simscenario1}
\end{figure}

When the cell-edge user location becomes close to either BS $3$ or BS $4$, i.e., $x \approx \pm \text{R}$, Fig.~\ref{Res_simscenario1} shows that using ZF-SIC RACMA at the nearest BS achieves the best detection performance among all other methods that use joint detection. This can be attributed to the fact that this user is no longer a ``common'' user - there is a large discrepancy among the received SNR at BS $3$ (very high) and at BSs $1,2,$ and $4$ (very low), and hence, the power imbalance severely affects the detection performance of the (G)CCA based approaches. This observation suggests that depending on the user's type (center or edge), one should use either the closest BS or the three nearest BSs to detect the user. In other words, if a user is relatively close to any BS, then this user's signal received power is high at this BS and very weak at all other BSs, and hence, it makes sense to decode this user's signal from the nearest BS. However, if a user is close to the edge between cells, then this user is {\it common} to multiple BSs and jointly detecting such a user from the three closest BSs using GCCA yields the best detection performance.      

More interestingly, it turns out that adding more BSs does not always improve the performance. For instance, at $x = \pm \text{R}/2$, while GCCA using four BSs attains a comparable performance relative to GCCA with the three nearest BSs, the latter is considerably better as the user moves towards the center $x = 0$. This is because at $x = 0$, the received SNR is very low at both BS $3$ and BS $4$ compared to BS $1$ and BS $2$, and consequently, both views $3$ and $4$ act as two ``noisy'' views that naturally degrade the signal recovery of the cell-edge user. Therefore, one can conclude that from the geometry of the  hexagonal cells, adding more BSs and feeding their received signals to GCCA to recover the common subspace will further degrade the detection performance of cell-edge users as any additional BS (view) will lead to an additional noisy view that severely affects the cell-edge user's signal recovery. In other words, the more views $(L  >  3)$ GCCA uses, the more difficult it becomes to reveal common information from all views simultaneously. Therefore, using the observation that using GCCA with the three closest BSs always yields the minimum BER for cell-edge users, and to further show the effectiveness of our approach under different settings, we will only consider the three BSs scenario, shown in Fig.~\ref{Simscenario2}, for all subsequent experiments. 

\begin{figure}
	\centering
	\includegraphics[width=2.35in]{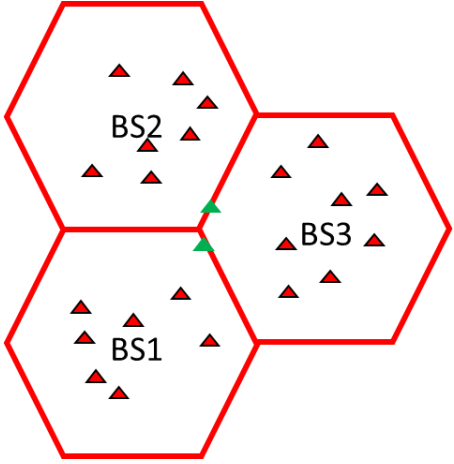}
	\caption{Snapshot from the three BSs simulated scenario}
	\label{Simscenario2}
\end{figure}

We now consider another experiment where we vary the transmitted power of the two cell-edge users from $ 20 $dBm to $ 25 $dBm which corresponds to approximately $0$dB to $5$dB SNR according to the adopted path-loss model,  while the transmitted power of all-center users is fixed to $25$dBm. Drawing different cell-center user locations for each Monte-Carlo realization, we compute the average BER among the two cell-edge users as a function of their transmitted power. Fig.~\ref{BER_SNR_All_Users} shows how GCCA provides significant improvement in the BER compared to all other methods. In particular, GCCA achieves an order of magnitude reduction in the BER compared to MMSE-SIC followed by RACMA (the second best method) which jointly detects the cell-edge user signals using the residual signals from the three BSs. Notice that, MMSE-SIC assumes perfect `oracle' knowledge of the channels of cell-center users at their serving BS. This assumption becomes less realistic when ``cell-center'' users are fully scattered throughout the cell. Although this gives a big advantage to the MMSE-SIC approach, GCCA still provides considerably better detection performance of cell-edge users. Furthermore, one can easily see how the channel estimation errors severely degrade the detection performance of cell-edge users.

	\begin{figure}
		\centering
		\includegraphics[width=3in]{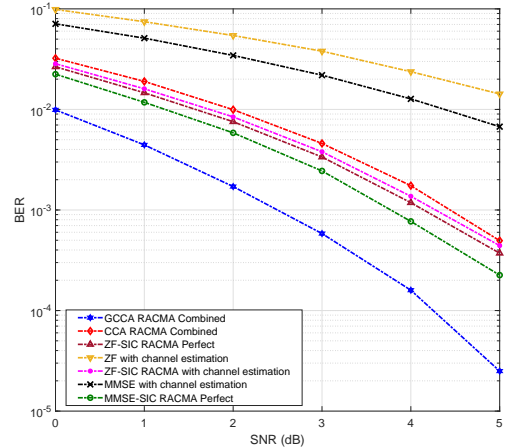}
		\caption{BER vs. SNR of cell-edge users}
		\label{BER_SNR_All_Users}
	\end{figure}%

	\begin{figure}
		\centering
		\includegraphics[width=3in]{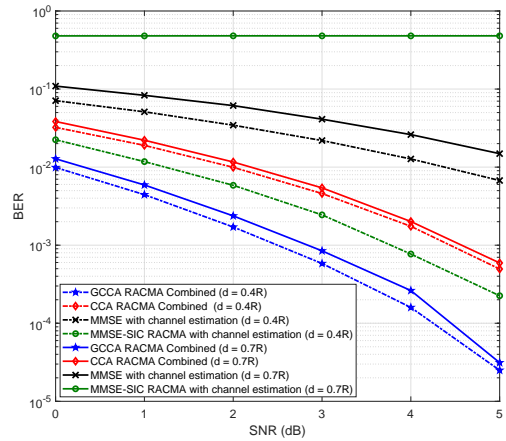}
		\caption{BER vs. SNR of cell-edge users, $d$ denotes the distance at which cell-center users are randomly dropped up to.}
		\label{Inter-cell interf}
	\end{figure}


Additionally, we simulate a more realistic scenario where cell-center users are almost fully scattered in their cell. In particular, cell-center users are dropped up to $d = 0.7 \text{R}$ from their serving BS, thereby cell-edge users are experiencing more aggressive inter-cell interference compared to the case where $d = 0.4 \text{R}$. Note that the MMSE-SIC RACMA is implemented using estimates of the cell-center users' channels instead of assuming perfect knowledge of their channels because that is hard to attain even approximately when cell-center users are fully scattered. Fig.~\ref{Inter-cell interf} depicts the inter-cell interference effect on the detection performance achieved by different methods. It is obvious that the MMSE-SIC RACMA completely fails at $d = 0.7 \text{R}$ compared to $d = 0.4 \text{R}$. On the other hand, both GCCA and CCA have a slight degradation in their performance, which in turn reflects the efficacy of both methods that principally rely on recovering the subspace of the ``equipowered'' users. Note that GCCA with three BSs still attains an outstanding detection performance compared to the other methods under this realistic scenario.      
	
\begin{figure}
	\centering
	\includegraphics[width=80mm,scale=2]{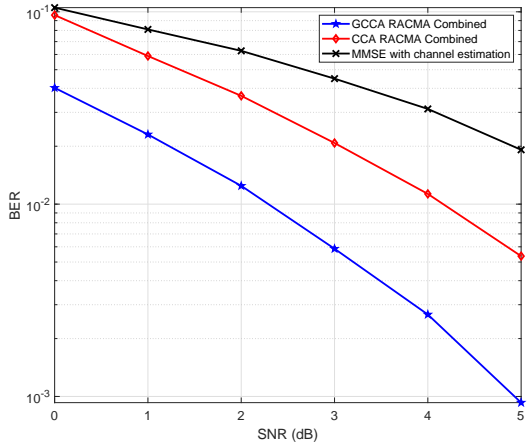}
	\caption{BER vs. SNR of cell-edge users for different number of users}
	\label{BER_SNR_users}
\end{figure}
We carry out another experiment in a more dense scenario where $K_\ell = 16$ and $M_\ell = 30$. Further, cell-center users are dropped up to $d = 0.8R$. We report the BER of cell-edge users versus the SNR. Although the results in Fig.~\ref{BER_SNR_users} show that the detection performance of all methods significantly degrades compared to the one in the previous experiment where $K_\ell = 8$, our proposed approach still can attain acceptable performance by achieving $1e-3$ BER at $5$dB. Notice that doubling the number of users and allowing them to be more scattered naturally leads to greater corruption in the estimated common subspace, and hence, the degradation in the detection performance obtained by the proposed method is expected. 

Adding more BSs with $K$ users each might be expected to severely affect the performance which is true in general. However since, in principle, our approach recovers the subspace containing the ``equipowered'' user signals, adding more users in the far cells (not served by the three closest BSs) can slightly affect the cell-edge detection performance as we observed from our simulations. For instance, looking at Fig.~\ref{Simu_Scenario1}, increasing the number of users served by BS $4$ does not affect the BER obtained when the cell-edge user is located at $x = -\text {R}/2$ if GCCA is used with BSs $1,2$ and $3$. However, increasing the number of cell-center users can only degrade the performance, as shown in Fig. \ref{BER_SNR_users}, when these users are served by any of the BSs used for detecting the cell-edge users via GCCA. This can be attributed to the fact that there is a chance that some users could be included in the common part, and hence, underestimating the dimension of the common subspace can degrade performance.
		
	\begin{figure}
		\centering
		\includegraphics[width=80mm,scale=2]{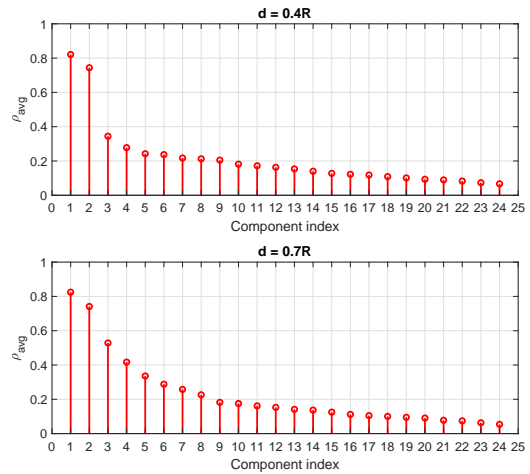}
		\caption{Average correlation coefficient of all possible extracted components via GCCA}
		\label{corr_comp}	 	
	\end{figure}

{Finally, we test the proposed algorithm used to detect the number of cell edge users (i.e., the common subspace dimension)}. We consider a setup with $M_\ell = 12$, $K_\ell = 8$, and $\text{SNR} \approx 3\text{dB}$. Note that since $2M_\ell < N$, we can find up to $2M_\ell$ canonical components. Fig.~\ref{corr_comp} shows the average correlation coefficient computed at the $i$-th extracted component, for $i = 1,\cdots,24$, for two different drop/scatter patterns for the cell-center users. It is obvious that when $d = 0.4\text{R}$, there is a significant gap between the average correlation coefficient of the first two components and the rest of the components. In particular, the average correlation coefficient of the first two components is almost $0.8$ while all the rest are less than $0.3$. Thus, one can decide that there exist only two cell-edge users in this case. On the other hand, at $d = 0.7\text{R}$, the value of the average correlation coefficient slightly increases for some of the components. For example, the average correlation coefficient of the third component now jumps to $0.52$ which means that there is one more user that can be considered as common user. However, considering only the first two components is enough to reliably recover the signals of the two cell-edge users as shown in Fig. ~\ref{Inter-cell interf}, where the detection performance was sightly affected by increasing $d$ from $0.4\text{R}$ to $0.7\text{R}$.     

\section{Conclusions}\label{Conc}
We studied the problem of cell-edge user signal detection in the uplink of a multi-cell, multi-user MIMO system, with the aim of designing a detector that can reliably demodulate cell-edge user signals in the presence of strong intra-cell  interference from cell-center users, without resorting to power control or/and scheduling algorithms that throttle the cell-center user rates. We proposed a GCCA-based approach that leverages {\em selective} base station cooperation to reliably identify the common subspace containing cell-edge user signals at low SNR, without even knowing their channels. Then, we used an efficient analytical method (RACMA) that guarantees the identifiability of binary signals from well-conditioned mixtures to separate the cell-edge user signals from the resulting subspace. 

We presented theoretical results to prove that under an idealized scenario, the proposed GCCA-based approach recovers the subspace containing the cell-edge user signals. Furthermore, we showed through an elegant analysis that under realistic assumptions on the inter-cell interference and the SNR of the cell-edge users, the common subspace recovery is guaranteed via GCCA. 
Simulations using a {\em realistic} propagation and system model were carried out to show the superiority of the proposed learning-based method over the prevailing state-of-the-art methods. In particular, our proposed approach attained an order of magnitude reduction in the BER compared to other multi-user detection methods that assume perfect knowledge of the channels of the cell-center users. Furthermore, our experimental results evaluated the cell-edge user detection performance as a function of the number of cooperating BSs, and revealed that using the three closest BSs is always optimal in this regard. This was not obvious {\em a priori}, as intuition may have suggested that two or even more than three BSs might be preferable in certain cases. 

\appendices
\section{Proof of Proposition 1}\label{AP}

In this section, we will show that the principal $K_c$ eigenvectors of the matrix ${\bf A} = \sum_{\ell = 1}^{L}{\bf Y}_\ell^H({\bf Y}_\ell{\bf Y}_\ell^H)^{-1}{\bf Y}_\ell$ is approximately the column space of the cell-edge user signals. We first write the auto-correlation matrix $ {\bf Y}_\ell{\bf Y}_\ell^H $ as
\begin{equation}\label{IC3}
	{\bf Y}_\ell{\bf Y}_\ell^H = {\bf Z}_{\ell}{\bf P}_\ell{\bf Z}_\ell^H + \sigma^2 {\bf I}       
\end{equation}
where we have exploited the facts that, at $ N > K_s $, $ {\bf X}^T{\bf X}/N \approx {\bf I}_{K_s}$ and $ \mathbb{E}[{\bf N}_\ell{\bf N}^H_\ell] = \sigma^2 {\bf I}_{N_\ell}$. 
Define the diagonal matrix $ {\bf \Gamma}_\ell := {\bf P}_\ell/\sigma^2 \in \mathbb{R}^{K_s \times K_s}$ that contains the received SNR of each user at the $ \ell $-th BS, and $ {\bf U}_\ell := {\bf N}^H_\ell/\sigma \in \mathbb{C}^{N \times M_\ell}$, and the matrix $ {\bf A}_\ell := {\bf Y}_\ell^H({\bf Y}_\ell{\bf Y}_\ell^H)^{-1}{\bf Y}_\ell $. Then, by direct substitution of~\eqref{IC2} and~\eqref{IC3} in the matrix $ {\bf A}_\ell $, we obtain
\begin{align}\label{IC4}
	\begin{aligned}
		{\bf A}_\ell &=  {\bf X}{\bf \Gamma}_\ell ^{1/2}{\bf Z}_\ell^H ({\bf Z}_{\ell}{\bf \Gamma}_\ell{\bf Z}_\ell^H + {\bf I})^{-1}{\bf Z}_{\ell}{\bf \Gamma}_\ell^{1/2}{\bf X}^T \\ 
		&\quad+ {\bf U}_\ell({\bf Z}_{\ell}{\bf \Gamma}_\ell{\bf Z}_\ell^H +  {\bf I})^{-1}{\bf U}^H_\ell + \boldsymbol{\delta}_\ell + \boldsymbol{\delta}_\ell^H
	\end{aligned}
\end{align}
where $\boldsymbol{\delta}_\ell := {\bf U}_\ell({\bf Z}_{\ell}{\bf \Gamma}_\ell{\bf Z}_\ell^H +  {\bf I})^{-1}{\bf Z}_{\ell}{\bf \Gamma}_\ell^{1/2}{\bf X}^T $. By applying the Woodbury matrix identity on the matrix $ {\bf C}_\ell := ({\bf Z}_{\ell}{\bf \Gamma}_\ell{\bf Z}_\ell^H +  {\bf I})^{-1} $, we get
\begin{subequations}\label{IC5}
	\begin{align}
		{\bf C}_{\ell} &=({\bf Z}_{\ell}{\bf \Gamma}_\ell{\bf Z}_\ell^H +  {\bf I})^{-1} \\
		&={\bf I} - {\bf Z}_\ell({\bf \Gamma_\ell}^{-1}+{\bf Z}_\ell^H{\bf Z}_\ell)^{-1}{\bf Z}_\ell^H\\
		&\approxeq  {\bf I}  - {\bf Z}_\ell{({\bf \Gamma_\ell}}^{-1}+{\bf I})^{-1}{\bf Z}_\ell^H\\
		&= {\bf I} - {\bf Z}_\ell{\bf D}_\ell{\bf Z}_\ell^H
	\end{align}
\end{subequations}
where $ {\bf D}_\ell \approxeq {\bf \Gamma}_\ell({\bf \Gamma}_\ell + {\bf I})^{-1} $.
It now follows that the first term in~\eqref{IC4} can be expressed as
\begin{subequations}\label{IC6}
	\begin{align}
		{\bf T}_{\ell}^{(1)} &= {\bf X}{\bf \Gamma}_\ell ^{1/2}{\bf Z}_\ell^H {\bf C}_\ell{\bf Z}_\ell{\bf \Gamma}_\ell^{1/2}{\bf X}^T \\
		&\approxeq {\bf X}{\bf \Gamma}_\ell ^{1/2}({\bf I}-{\bf D}_\ell){\bf \Gamma}^{1/2}_\ell{\bf X}^T  \\
		&= {\bf X}{\bf D}_\ell{\bf X}^T
	\end{align}
\end{subequations}
On the other hand, the second term in~\eqref{IC4} can be written as
\begin{subequations}\label{IC}
	\begin{align}
		{\bf T}_{\ell}^{(2)} &= {\bf U}_\ell{\bf C}_\ell{\bf U}^H_\ell \\
		&= {\bf U}_\ell{\bf U}^H_\ell - {\bf U}_\ell{\bf Z}_\ell{\bf D}_\ell{\bf Z}^H_\ell{\bf U}^H_\ell 
	\end{align}
\end{subequations}
Given that $ {\bf U}_\ell $ contains i.i.d entries with zero mean and variance $ 1/N $ while $ {\bf Z}_\ell $ contains i.i.d entries with zero mean and variance $ 1/M_\ell $,  both $ {\bf Z}_\ell  $ and $ {\bf U}_\ell $ are uncorrelated, and ${\bf D}_\ell \preceq {\bf I}_{K_s} $ ($ \preceq $ interpreted element-wise), then it follows that the summation in~(\ref{IC}b) will be dominated by the matrix ${\bf U}_\ell{\bf U}^{H}_\ell$. Therefore, $ {\bf T}_{\ell}^{(2)} $ can be approximately written as
\begin{align}
	\label{IC7}
	{\bf T}_{\ell}^{(2)} \approxeq {\bf U}_\ell{\bf U}^H_\ell
\end{align} 
Then, the expression $\boldsymbol{\delta}_\ell$ can written as
\begin{subequations}\label{ICC}
	\begin{align}
		\boldsymbol{\delta}_\ell & = {\bf U}_\ell {\bf C}_\ell {\bf Z}_\ell {\bf \Gamma}_\ell^{1/2}{\bf X}^T\\
		&\approxeq {\bf U}_\ell {\bf Z}_\ell({\bf I} - {\bf D}_\ell) {\bf \Gamma}_\ell^{1/2}{\bf X}^T\\
		& = {\bf U}_\ell {\bf Z}_\ell (\boldsymbol{\Gamma}_\ell + {\bf I})^{-1} \boldsymbol{\Gamma}_\ell^{1/2} {\bf X}^T
	\end{align}
\end{subequations}
By summing~\eqref{IC7} and (\ref{ICC}c), we get
\begin{align}\label{IC13}
	{\bf T}_{\ell}^{(2)} + \boldsymbol{\delta}_\ell & = {\bf U}_\ell ({\bf U}^H_\ell+{\bf Z}_\ell (\boldsymbol{\Gamma}_\ell + {\bf I})^{-1} \boldsymbol{\Gamma}_\ell^{1/2} {\bf X}^T)
\end{align}
where the summation on the right hand side of~\eqref{IC13} is nothing but adding two Gaussian matrices; one with variance $ \sigma_1^2 = 1/N $ and the other with variance $ \sigma_2^2 = \frac{1}{N}\frac{1}{M_\ell}\sum_{i = 1}^{K_s} \frac{\sqrt{\gamma_{i\ell}}}{(\gamma_{i\ell}+1)^2} $, where it can be easily seen that, even for modest $M_\ell $, $ \sigma^2_2 << \sigma^2_1$. Therefore, the summation in~\eqref{IC13} will be dominated by $ {\bf T}_{\ell}^{(2)} $. Thus, combining~\eqref{IC6} with~\eqref{IC7},~\eqref{IC4} can be written as
\begin{equation}\label{IC8}
	{\bf A}_\ell \approxeq {\bf X}{\bf D}_\ell{\bf X}^T + {\bf U}_\ell{\bf U}^H_\ell
\end{equation}
Recall that the optimal solution $ {\bf G}^* $ of \eqref{MAXVAR} is the $ K_c $ principal eigenvectors of the following matrix
\begin{subequations}\label{IC9}
	\begin{align}
		{\bf A} &= \sum_{\ell = 1}^L {\bf A}_\ell\\
		&= {\bf X}{\bf D}{\bf X}^T + \sum_{\ell = 1}^{L} {\bf U}_\ell{\bf U}^H_\ell
	\end{align}
\end{subequations}
where $ {\bf D }:= \sum_{\ell = 1}^{L}{\bf D}_\ell \in \mathbb{R}^{K_s \times K_s}$. By defining $ {\bf V}= [{\bf X},{\bf U}_1,\cdots,{\bf U}_L] \in \mathbb{C}^{N \times (Ks+\sum_{\ell=1}^{L}M_\ell)}$ and $ {\bf \Sigma} := \text{Diag}({\bf D},{\bf I}_{M_1},\cdots,{\bf I}_{M_L}) $,~\eqref{IC9} can be equivalently expressed as
\begin{equation}\label{IC10}
	{\bf A} = {\bf V \Sigma V}^H 
\end{equation}
Since $ {\bf X}^T{\bf X} \approx {\bf I}_{K_s} $ and by definition $ {\bf U}^H_\ell{\bf U}_\ell \approx {\bf I}_{M_\ell} \forall \ell $, then $ {\bf V}^H{\bf V} \approx {\bf I}_{K_s + M_s} $, it can be readily seen that the right hand side of~\eqref{IC10} is nothing but the eigendecomposition of the matrix $ {\bf A} $. Recall that the $ i $-th diagonal entry of the matrix $ {\bf D} $ is given by
\begin{equation}\label{IC11}
	{\bf D}_{(i,i)} = \sum_{\ell = 1}^{L} r_{i\ell}
\end{equation}
From (AS2) and by assuming that the received signal power of the $k$-th cell-edge user at the $\ell$-th BS is few dBs above the noise floor, i.e., $r_{k\ell} > 0.5~\forall k =1,\cdots,K_c~\text{and}~\ell = 1,2,3$, the eigenspace of the $ K_c $ principal components of the matrix $ {\bf A} $ is given by 
\begin{equation}\label{IC12}
	{\bf G}^{*} = {\bf X}_c{\bf P}
\end{equation}
where ${\bf P}$ is any $K_c \times K_c$ non-singular matrix.

\bibliographystyle{IEEEtran}
\bibliography{IEEEabrv,refrences}

\end{document}